\documentclass[letterpaper, 10pt, conference]{ieeeconf}
\IEEEoverridecommandlockouts  
\usepackage{amssymb}
\usepackage{amsfonts} \usepackage{amsmath} \usepackage{bm}
\usepackage{latexsym} \usepackage{epsfig}
\usepackage{hyperref}
\usepackage{cuted}
\usepackage{authblk}

\newtheorem{theorem}{Theorem}[section]
\newtheorem{lemma}[theorem]{Lemma}

\newtheorem{definition}[theorem]{Definition}

\newtheorem{conjecture}[theorem]{Conjecture}
\newtheorem{open}[theorem]{Open Problem}

\newcommand{\ignore}[1]{}

\newcommand{\enote}[1]{} \newcommand{\knote}[1]{}
\newcommand{\rnote}[1]{}

\newcommand{\CondE}[2]{{\mathbb{E}}\left[{#1}\middle\vert{#2}\right]}

 \newcommand{\R}{\mathbb R}

\newcommand{\half}{{\textstyle \frac12}}

\renewcommand{\phi}{\varphi}

\begin{document}
\title{Efficient Bayesian Learning in Social Networks with Gaussian
  Estimators}
\author[1]{Elchanan Mossel}
\author[2]{Noah Olsman}
\author[3]{Omer Tamuz}

\affil[1]{University of Pennsylvania and U.C. Berkeley \authorcr Email: {\tt mossel@stat.berkeley.edu}\vspace{1.5ex}}
\affil[2]{California Institute of Technology \authorcr Email: {\tt nolsman@caltech.edu}\vspace{1.5ex}} 
\affil[3]{California Institute of Technology \authorcr Email: {\tt tamuz@caltech.edu} } 

\date{\today}
\maketitle
\begin{abstract}
  We consider a group of Bayesian agents who try to estimate a state
  of the world $\theta$ through interaction on a social network. Each agent $v$
  initially receives a private measurement of $\theta$: a number $S_v$
  picked from a Gaussian distribution with mean $\theta$ and standard
  deviation one.  Then, in each discrete time iteration, each reveals
  its estimate of $\theta$ to its neighbors, and, observing its neighbors'
  actions, updates its belief using Bayes' Law.

  This process
  aggregates information efficiently, in the sense that all the agents converge to the
  belief that they would have, had they access to all the private
  measurements. We show that  this process is computationally efficient, so
  that each agent's calculation can be easily carried out. We also
  show that on any graph the process converges after at most $2N \cdot D$ steps,
  where $N$ is the number of agents and $D$ is the diameter of the
  network. Finally, we show that on trees and on distance transitive-graphs the process converges after $D$ steps, and that it preserves privacy, so that agents learn very little about the private signal of most other agents, despite the efficient aggregation of information.
  Our results extend those in an unpublished manuscript of 
the first and last authors.
\end{abstract}

\section{Introduction}
We study a model of social learning, in which a group of
Bayesian agents learn a state of the world through repeated
interaction with their social network neighbors\footnote{This work is an extension of an unpublished manuscript by Mossel and Tamuz~\cite{mossel2010efficient}}.

Similar models which have been studied in the past can be roughly
divided into two categories: {\em rational} models and {\em rule of
  thumb} models. In rational models (e.g.\ Gale and
Kariv~\cite{GaleKariv:03}, Rosenberg, Solan and
Vieille~\cite{rosenberg2009informational}, Mossel, Sly and Tamuz~\cite{mossel2015strategic}, Arieli and Mueller-Frank~\cite{arieli2014inferring}), agents choose actions that
are optimal under some criterion; usually the maximization of some
expected utility.  In rule of thumb models (e.g.\
DeGroot~\cite{DeGroot:74}, Bala and Goyal~\cite{BalaGoyal:96}, Golub
and Jackson~\cite{golub2010naive}) they act by some fixed heuristic.

Rational models are more natural and conform to the economic paradigm
of rational agents, making them amenable to game theoretical
analysis. However, the calculations that are required of the agents
there are usually complicated and perhaps computationally hard. In
this paper we present a rational model for which, as we show, the
agents' actions can be calculated efficiently. An essentially
identical model was studied by DeMarzo, Vayanos and
Zweibel~\cite{DeMarzo:03}, who, however, did not consider
computational questions.

Beyond computation efficiency we study the speed of convergence of the learning process. Denoting the number of agents by $N$ and the diameter of the social network graph by $D$, we show that on any graph the process converges after at most $2N \cdot D$ steps. On trees and on distance-transitive graphs (e.g., the hypercube) we show that the process converges after $D$ steps; note that $D$ is a lower bound on convergence time in any graph.

DeMarzo, Vayanos and Zweibel~\cite{DeMarzo:03} show that information is optimally aggregated by this process: all agents eventually converge to the same estimate they would have if they shared all the private signals. We develop a notion of privacy, and show that despite the fact that information is aggregated, on some graphs a high degree of privacy is preserved, in the sense that most pairs of agents know very little about each other's private signals.

\subsection{Model}

We consider a finite set of Bayesian agents $V$ and denote
$N=|V|$. The agents are connected by an undirected {\em social network graph} $G=(V,E)$. We assume that the graph is connected.

The agents are interested in estimate a {\em state of the world} $\theta
\in \R$. Initially, each agent $v \in V$ receives a {\em private
  signal} $S_v$, drawn from the normal distribution with mean $\theta$ and
standard deviation $\sigma$. We assume that the different $S_v$ are independent.

The agents initially have some prior belief regarding $\theta$, and update
it to a posterior belief, according to Bayes' Law, with each
additional piece of information they encounter. Both prior and
posterior beliefs are distributions on the possible values of $\theta$.  We
assume that all agents share a common prior, the ``improper'' uniform
measure on $\R$. An equivalent model would let each agent $v$ have a
prior equal to the Gaussian distribution with mean $\theta$ and variance
one.

At each iteration $t$, each agent $v$ reveals to its neighbors the
expectation of its current belief $X_v(t)$, and learns theirs. It then
updates its belief, based on this new information:
\begin{align*}
  X_v(t) = \CondE{\theta}{S_v,\{X_w(t')\,:\,(v,w) \in E,t'<t\}}.
\end{align*}
Note that to perform this calculation each agent has to know the
structure of the graph.

This model is similar to the one presented in
\cite{MosselTamuzAllerton:09}.  The agents in that model, however,
were not Bayesian and had no memory of their observations in past
iterations.

We note that this model can be trivially generalized to a much wider
class of private signal distributions, provided that we assume that
the agents aim to minimize square error. In particular, in that case
all of our results hold for any private signal distribution in which
the covariances of the private signals (of each pair of agents) are
finite.

\section{Results}
We prove the following results:

\textbf{Efficient Computation.}  Each agent's calculation of $X_v(t)$ is
computationally efficient: it can be achieved using simple linear
algebra operations, involving matrices whose size is the size of the
network.
    
{\bf Efficient and Rapid Learning.} DeMarzo, Vayanos and
Zweibel~\cite{DeMarzo:03} show that the agents' posterior beliefs all
converge to the same value, $\CondE{\theta}{\{S_v\}_{v \in V}}$.  This is
the value that they would have converged to, had they all access to
each others' private measurements.  We show that the process converges
in at most $2N \cdot D$ iterations, where $N$ is the number of agents and $D$ is the diameter of the graph.

\textbf{Network Topology and Optimal Convergence.} We prove that certain network topologies permit convergence in $D$ steps, the fastest possible convergence time. Specifically, we show that networks whose underlying graph is distance-transitive (e.g. hypercube graphs and Johnson graphs) or is a tree have optimal convergence time, largely due to a high degree of symmetry that can be exploited by each agent when calculating their estimates.

\textbf{Privacy.} We develop a notion of privacy given each agent's information: the privacy between agent $v$ and $w$ is the variance of $v$'s best estimate of $w$'s private signal, at the end of the process. We derive a simple analytic expression for privacy, and show that on trees and on distance-transitive graphs a high degree of privacy is preserved: for most pairs of agents $v$ and $w$, $v$'s estimate of $w$'s private signal at the end of the process is not much more precise than it was in the beginning.

\subsection{Computational efficiency.}
We choose a model of computation in which agents can store real
numbers and carry out the basic arithmetic operations on them. A
feature of this computational model is that it circumvents such issues
as numerical stability, for example in the inversion of
ill-conditioned matrices. This modeling approach makes it easier to
construct efficient algorithms, but is arguably less realistic than a
model that takes these issues into account.

\subsubsection{The agents' calculation}

Let $\mathcal{W}$ be the vector space of Gaussian random variables spanned by the different $S_v$'s:
\begin{equation}
  \label{eq:rv_vector_space}
  \mathcal{W}=\left\{\sum_{v\in V}\beta_v S_v\mbox{ s.t. } \forall v:\beta_v \in \R\right\}.
\end{equation}
It is easy to convince oneself that this indeed is a vector space of finite dimension. Note that all the random variables in this space are normally distributed. Denote by $\mathcal{W}^1$ the subset of unbiased estimators of $\theta$ in 
$\mathcal{W}$: 
\begin{equation}
\mathcal{W}^1 = \left\{ \sum_{v\in V}\beta_v S_v\in \mathcal{W} \mbox{ s.t.} \sum_{v\in V}\beta_v = 1 \right\}.
\end{equation}

\begin{theorem}
  For all agents $w$ and times $t$, it holds that $X_w(t)\in\mathcal{W}^1$,     with
  $X_w(t) = \sum_v\beta_{wv}(t) S_v$ for some $\beta_{wv}(t)$. 
\end{theorem}
\begin{proof}
  We shall prove this by induction on $t$. At time $t=0$ the claim is
  true since $\beta_{wv}(0)$ is one when $w=v$ and zero
  otherwise. Assume that the claim is true until time $t$.

  Consider an agent $w$, and denote by $r_0,\ldots,r_k$ the random
  variables that agent $w$ has observed up to time $t$, with
  $r_0=S_w=X_w(0)$. Those are $w$'s own and its neighbors' past
  estimators.  By our assumption these are all in $\mathcal{W}^1$, and we
  can write $r_i=\sum_v A_{iv}S_v = (AS)_i$, where the coefficients of the matrix  $A$ are a simple re-indexing of the coefficients
  $\beta_{wv}(t)$, by some relation that maps each $w$ and $t$ to some
  $i$. Since by assumption $r_i\in \mathcal{W}^1$ then
  $\sum_v A_{iv}=1$.

  Denote by ${\bf r}$ the
  vector $(r_0,\ldots,r_n)$, denote by ${\bf 1}$ the
  vector $(1,\ldots,1)\in \R^n$, and denote by $C_{ij}$ the covariance between 
  $r_i$ and $r_j$, so that
  \begin{equation}
    \label{eq:C}
    C = \sigma^2 A A^T. 
  \end{equation}
  Then ${\bf r}$'s distribution is the normal
  multivariate distribution with covariance matrix $C$ and mean 
  ${\bf 1}\theta$ (since $r_i\in\mathcal{W}^1$), and
  the likelihood of $\theta$ given that agent $w$ has observed  ${\bf r}$ is
  \begin{equation}
    \label{eq:s_likelihood}
    p({\bf r}|\theta)=\frac{1}{\left(2\pi\right)^{n/2}|C|^{1/2}}
    e^{-\half({\bf r}-{\bf 1}\theta)^T C^{-1}({\bf r}-{\bf 1}\theta)},
  \end{equation}
  where $p(\cdot)$ denotes probabilities in our probability space.
  Note that in the case that $C$ is not invertible
  (equivalently, ${\bf r}$ is not linearly independent) we remove from
  it (and correspondingly from ${\bf r}$) a minimal set of columns and
  rows such that it becomes invertible. By corollary, $C$ is
  never larger than $N \times N$.

  The expression $({\bf r}-{\bf 1}\theta)^T C ^{-1}({\bf r}-{\bf 1}\theta)$ can
  be rewritten as
  \begin{equation*}
    {\bf 1}^T C^{-1}{\bf 1}\cdot\left(S - \frac{ {\bf 1}^T{\bf C}^{-1}{\bf r}}
        {{\bf 1}^T C^{-1} {\bf 1}}\right)^2 + B
  \end{equation*}
  with $B$ a normalization factor. Denote
  \begin{equation}
    \label{eq:gamma}
    {\bf \gamma} = \frac{ {\bf 1}^T C^{-1}}{{\bf 1}^T C^{-1}{\bf 1}},
  \end{equation}
  and note that $\sum_i\gamma_i=1$.

  We can now write
  \begin{equation}
    \label{eq:posterior_explicit}
    p({\bf r}|S)=\frac{1}{\sqrt{2\pi\tau^2}}e^{-(\theta-x)^2/2\tau^2}
  \end{equation}
  where
  \begin{equation}
    \label{eq:new_estimator}
    x=       \frac{ {\bf 1}^T  C^{-1}}
      {{\bf 1}^T C^{-1}{\bf 1}
    }
    {\bf r}=\sum_i\gamma_ir_i \quad\quad\mbox{and}\quad\quad \tau^2 = \frac{1}{{\bf 1}^T C^{-1}{\bf 1}}\cdot
  \end{equation}
  Note that $x$ is a linear combination of the observations that $w$
  made up to time $t$.

  The expected value of the multinormal
  distribution~\eqref{eq:posterior_explicit} is $x$, and therefore the
  maximum likelihood estimator is $x$. Since the prior is uniform the
  Bayes estimator is likewise $x$, and we have that $X_w(t+1)=x$.
  Then
  \begin{equation}
    X_w(t+1)=\sum_i\gamma_i\sum_{v}A_{iv}S_v
  \end{equation}
  and therefore 
  \begin{equation}
    \label{eq:new_beta}
    \beta_{wv}(t+1)=\sum_i\sum_{v}\gamma_i A_{iv}.     
  \end{equation}
  Since $\sum_i\gamma_i=1$ and $\sum_v A_{iv}=1$ then $\sum_v\beta_{wv}(t+1)=1$. We have shown then that $X_w(t+1)\in\mathcal{W}^1.$ We have also shown that to calculate $\beta_{wv}(t+1)$, given the  coefficients at time $t$, one need only invert $N$ matrices (one for each agent), of size at most $N \times N$ - certainly an efficient calculation. Furthermore, no knowledge of the $S_v$'s is needed, but only of the graph structure.
\end{proof}

We write below an algorithm that efficiently calculates all the
vectors $\gamma$ for all the agents at all time periods. Given this,
an agent can straightforwardly calculate its actions
using~\eqref{eq:new_estimator}.

We use here the notation introduced in the proof above, but add to it
explicitly the name of the calculating agent and the time
period. Hence $A^w_{iv}(t)$ is the $A_{iv}$ of agent $w$ at
time $t$, and likewise for $C,\beta$,$\gamma$ and $\tau$.
\begin{enumerate}
\item Calculate all $\beta_{wv}(t)$. At the first time period these
  are trivial, as we note above. At later time periods these are the
  result of the calculation of the previous time periods.
\item Calculate all $A^w_{iv}(t)$, as described above. This is a
  simple renaming of $\beta_{wv}(t)$.
\item Calculate all $C^w(t)$ by~\eqref{eq:C}.
\item Calculate all ${\bf \gamma}^w(t)$ by~\eqref{eq:gamma}.
\item Calculate all $\beta_{wv}(t+1)$ by~\eqref{eq:new_beta}.
\end{enumerate}

\begin{theorem}
  There exists an efficient algorithm to calculate the agents' actions.
\end{theorem}
\begin{proof}
  To calculate its action at time $t+1$, an agent $w$ can run the
  algorithm above up to time $t$ to calculate ${\bf \gamma}^w(t)$, and
  then use~\eqref{eq:new_estimator} to calculate $X_w(t+1)$. The
  running time is dominated by the inversion of the covariance
  matrices $C$, and hence is a small polynomial in $n$ and
  linear in $t$.
\end{proof}

\subsection{Learning efficiency}
\subsubsection{Convergence in $N^2$}

To show that the beliefs of the agents converge, we need only note
that being conditional probabilities over increasingly large
probability spaces, these beliefs are martingales. Then, because these
martingales are bounded in $L^2$, they converge.  However, the
following proof, which does not require the power of martingales,
shows that convergence in fact takes places in at most $N^2$
iterations, and that furthermore all agents converge to the same
belief. This result is in spirit related to previous martingale
results in social learning such as Aumann's ``Agree to Disagree''
result~\cite{Aumann:76, Geanakoplos:92} and the work of Borkar and
Varaiya~\cite{BorVar:82}.  The proof is similar to the one presented
by DeMarzo et al.\ \cite{DeMarzo:03}.

When two neighboring agents have different beliefs, then at least one
of them will learn from the other and improve its estimator: Assume
agents $u$ and $v$ are neighbors with different estimators, and agent
$v$'s belief has variance lower than or equal to that of agent
$u$. Then agent $v$'s estimator is necessarily not in the space
spanned by the estimators previously seen by $u$. Hence the dimension
spanned by $u$'s memory will increase at this iteration. We have thus
shown that in each iteration, unless all the agents have the same
estimator, at least one of them increases the dimension of its space
by at least one. Since the maximum dimension possible is $N$ then
convergence will occur after at most $N^2$ steps, and all agents will
converge to the same belief.

\subsubsection{Convergence in $2 N \cdot D$ iterations}
A slightly more subtle argument proves a better bound for the convergence
rate, namely $2 N\cdot D$, where $D$ is the diameter of the graph. The idea
of the proof is that the current estimator of an agent $u$ cannot remain
unchanged for many steps, unless a growing neighborhood around $u$ also remains
stagnant. The formal proof uses the following lemma.

\begin{lemma}
  If some agent's estimator has not changed for $2 D$ steps then the
process has converged.
\end{lemma}
\begin{proof}
  Assume agent $u$'s estimator does not change from iteration $t_0$ to $t_0 + 2D$, so that
  \[
  X_u(t_0)=X_u(t_0+1)=\cdots=X_u(t_0+2D). 
  \] 
Denote $x:=X_u(t_0)=\cdots=X_u(t_0+2D)$,
and let $\mathcal{U}$ be the space spanned by the
estimators in $u$'s memory at time $t_0 + 2D$. Then by definition of the process
$x$ is the optimal unbiased estimator in $\mathcal{U}$.

Let $w$ be a neighbor of $u$. Then $w$'s estimator at time $t_0+1$, $X_w(t_0+1)$,
is in $\mathcal{U}$, since $u$ observes $X_w(t_0+1)$ at time $t_0+2$. Now $x$ by
definition is better than any estimator in $\mathcal{U}$, and so, since
$w$ has observed $x$ at time $t_0$, it must be that $X_w(t_0+1)=x$. By the
same argument $X_w(t)=x$ for $t_0+1\leq t\leq t_0+2D-1$.

Applying this argument inductively, it follows that at time
 $t_0+i \leq t \leq t_0+2D-i$ all
the agents at distance $i$ from $u$ have estimator $x$, and so at time $t_0+D$
all agents have the same estimator. 

Recalling that at each iteration an agent's estimator is a weighted
average of those of its neighbors, we conclude that all nodes will
have estimator $x$ for all times $t \geq t_0+D$.  The proof follows.

\end{proof}

\begin{theorem}
  The process stops after $2 N \cdot D$ iterations.
\end{theorem}
\begin{proof}
  Every time an agent's estimator changes, the dimension of the space the
agent's memory spans increases by at least one, and so this cannot happen more
than $N$ times. Since This must happen every $2 D $ steps as long as the 
process hasn't converged,
the process must stop after $2 N \cdot D$ iterations.
\end{proof}

\subsubsection{Convergence to the optimum}
At any particular iteration, any node $v$ contains $S_v$ in the space of its estimators.
At each iteration the estimator at $v$ is then of the form $a S_v + b S$ where $S$ is an unbiased linear estimator based on some signals (not including $S_v$), and $a+b=1$. Note that the variance of this estimator is $a^2 + b^2 Var(S)$ and it is minimized when $a = Var(S)/(1+Var(S))$.
Since $S$ depends on all the signals but $S_v$ its variance is at least $1/(N-1)$ and therefore $a$ is at least $1/N$.

Hence all the agents, at all iterations, give their own estimators
weight which is at least $1/N$. Since they all converge to the same
estimator, and since the sum of the weights in this estimator must be
one (since it, too, is unbiased), then the weights must all be $1/N$,
and the limiting estimator is the simple average of the private measurements, as stated.

\section{Diameter Convergence}
	We say that two agents $v$ and $w$ have observed each other at time $t$ if $d(v,w) \leq t$. Fundamentally, the fastest possible convergence time will occur in a number of steps equal to the diameter $D$ of $G$. This is because the optimal estimator requires each agent to average over all of the private signals, and there must have been $t > D$ steps before the most distant pair of agents can observe each other. Here we show that at least two families of graphs achieve this optimal convergence. We start by proving a general lemma that will be useful in this analysis.
	\begin{lemma}\label{lem:diam}
		If each agent's calculation corresponds to averaging all the $S_w$ it has observed so far, e.g. all of the $\beta_wv$ are either equal to $0$ if $d(v,w) > t$ and $\frac{1}{m}$ for some integer $m> 0$ if $d(v,w) \leq t$, then all beliefs converge after $D$ steps. 
	\end{lemma}
	\begin{proof}
		By definition, there exist a pair of agents $v$ and $w$ between whom $d(v,w) = D$. This means that, after $D$ steps, $w$ will observe an action for which $\beta_{wv} > 0$, and vice versa. Since by assumption each agent's calculation is an average of the all $S_w$ observed so far, and $v$ and $w$ both have weights $\beta_{wv} \neq 0$, it must be true that $v$ and $w$ will both have beliefs which have converged (e.g. $\beta_{wv} = \frac{1}{N}$). Since $v$ and $w$ were the most distant pair of agents, this implies that all other agents have beliefs that have converged as well.
	\end{proof}
	
	This lemma allows us to reduce the diameter convergence proof to simply showing inductively that, if all agents compute uniform averages at step $t$, then they will also do so at step $t+1$. We use this approach to prove that two families for graphs, trees and hypercubes, converge in precisely $D$ steps. In general we find that all distance-transitive graphs (e.g. hypercubes, Johnson Graphs, etc.) have this convergence property.
	
	\begin{theorem}[Trees]\label{thm:trees}
		Beliefs converge in $D$ steps if $G$ is a tree.
	\end{theorem} 
	\begin{proof}
		We start by selecting an arbitrary agent $w$ to be the root. It is straightforward to see that at $t=1$, agent $w$ simply averages its own observation $S_w$ with those of each of its children $c \in C_w$. Now assume that, up to time $t$, each agent's current estimator is the average of all the $S_v$'s it has had access to so far, e.g. 
		$$X_w(t) = \sum_{v \in B_i(t)}A_{wv}S_v = \frac{1}{|B_w(t)|}\sum_{v \in B_w(t)}S_v,$$
		where $B_w(t)$ is the number of agents within graph distance $t$ of agent $w$. This means that, at time $t+1$, $w$ has access to its childrens' observations $X_{c}(t)$ for $c \in C_w$, as well as all of its previous observations $x_w(t')$ for $t' \leq t$. We define $G_{c}(t)$ to be the subtree of the graph $G$ corresponding to child $c$ up to depth $t$. Similarly, let $B_{c}(t)$ be the set of agents in a ball of radius $t$ that each child $c$ has observed up to time $t$.
		
		If we think about what information is encoded in $X_{c}(t)$, we know that the agents whose private observations $S_v$ have non-zero weights in this estimate will be those in $G_{c}(t+1)$, agent $w$ itself, and agents in $G_{c'}(t-1)$ (with $c' \neq c$). This means that each $X_{c}(t)$ contains duplicate observations in $B_w(t-1)$, and non-overlapping observations of all agents radius exactly $t+1$ from $w$. This means that $w$ can calculate 
		\small\begin{multline}\label{eq:tree}
			\sum_{c \in C_w}\frac{|B_{c}(t)|}{|B_w(t+1)|} X_{c}(t) 
          \\  - (|C_w|-1)\frac{|B_w(t-1)|}{|B_w(t+1)|} X_w(t-1) 
           \\ = \frac{1}{|B_w(t+1)|}\sum_{v \in B_w(t+1)} S_v.
		\end{multline}\normalsize

		Since $\frac{1}{|B_w(t+1)|}\sum_{v \in B_w(t+1)} S_v$ is the optimal estimate of all agents that $w$ has observed at time $t$, it must be true that
        $$X_w(t+1) = \frac{1}{|B_w(t+1)|}\sum_{v \in B_w(t+1)} S_v.$$
        Since $w$ was an arbitrary agent, these results must hold for all agents. Since the inductive step holds, we see that each agent's calculation corresponds to a uniform average at each step and by lemma \ref{lem:diam}, beliefs converge in $D$ for the case of tree graphs.
	\end{proof}
	One straightforward consequence of~\eqref{eq:tree} is that if $w$ is a leaf, then it simply copies its neighbor at each time step because $|C_w| = 1$ and $|B_w(t+1)| = |B_c(t)|$.
	\begin{theorem}[Hypercubes]\label{thm:hcube}
		Beliefs converge in $D$ steps if $G$ is a hypercube.
	\end{theorem}
	\begin{proof}
		To begin we index agents on an $n$-dimensional hypercube with vectors in $\mathbb{F}_2^n$. Edges connect Hamming distance 1 agents. More generally, the distance of two agents in the graph is the Hamming distance of their indices. We proceed inductively in much the same way as the previous theorem. Without loss of generality, we pick an agent to index as $\textbf{0}$, the vector of all zeros whose belief we denote as $X_0$. Agent $\textbf{0}$'s $i$th neighbor can then be indexed by $\delta_i$, the vectors whose $i$th entry is $1$ and all other entries are $0$ whose belief we denote as $x_{\delta_i}(t)$.
		
		Again we see that the base case is simple, as each agent averages themselves with their neighbors at $t=1$. We assume that, up to time $t$, each agent's estimate is the average of all agents in a ball of radius $t$ around it,
		$$X_w(t) = \frac{1}{|B_w(t)|}\sum_{v \in B_w(t)}S_v, \forall w \in \mathbb{F}_2^n.$$
		At time $t+1$, agent $\textbf{0}$ receives beliefs from all $n$ of its neighbor at time $t$, $X_{\delta_i}(t)$. We first note that these beliefs can be decomposed in the following way,
		\begin{multline}\label{eq:decomp}
       		X_{\delta_i}(t) = \frac{|B_0(t-1)|}{|B_{\delta_i}(t)|}X_{0}(t-1) \\+ \frac{1}{|B_{\delta_i}(t)|}\sum_{v \in H_{\delta_i}(t-1) \cup H_{\delta_i}(t)}S_v,
        \end{multline}
		where $H_w(t)$ represents the hull of agents at distance exactly $t$ from agent $w$. While in the case of trees we took advantage of the non-overlapping nature of agents at radius $t+1$ in distinct subtrees, for hypercubes we do not necessarily have that $H_{\delta_i}(t) \cap H_{\delta_j}(t) = 0$ for $i \neq j$. This will generally not be the case after $t = 0$, so to compute a uniform average we must make sure that each agent is in exactly the same number of hulls $H_{\delta_i}(t)$.
		
		To prove this, we note that agents at radius $t+1$ have indices corresponding precisely to vectors in $\mathbb{F}^n_2$ with exactly $t+1$ non-zero entries. From this we know that each agent at radius $t+1$ will be accounted for in exactly $t+1$ of the neighbors $\delta_i$, specifically in every $\delta_i$ corresponding to a non-zero entry of that agent's index. The same will be true for agents at radius $t$. Using equation \eqref{eq:decomp}, we see that
	\small\begin{eqnarray*}
        \sum_{i = 1}^n \frac{|B_{\delta_i}(t)|}{|B_0(t+1)|} x_{\delta_i}(t) 		&=&  (n-t)\frac{|B_0(t-1)|}{|B_0(t+1)|}X_0(t-1) \\ 
        &-& \frac{|B_0(t)|}{|B_0(t+1)|}X_0(t) \\ 
        &+&  \frac{t+1}{|B_0(t+1)|}\sum_{w \in B_0(t+1)}S_w.
	\end{eqnarray*}\normalsize

	Rearranging, we have that 
\small\begin{multline}\label{eq:hcube}
  \frac{1}{|B_0(t+1)|}\sum_{w \in B_0(t+1)}S_w 
  		\\ = \frac{1}{t+1} \Bigg( \sum_{i = 1}^n \frac{|B_{\delta_i}(t)|}{|B_0(t+1)|} X_{\delta_i}(t) \\
        + \frac{|B_0(t)|}{|B_0(t+1)|}X_0(t) \\
        - (n-t)\frac{|B_0(t-1)|}{|B_0(t+1)|}X_0(t-1) \Bigg).
\end{multline}\normalsize
Now we can again use lemma \ref{lem:diam} to prove that, if $G$ is a hypercube, then the process converges in $D$ steps.
\end{proof}

Note that the only property of the hypercube we used in Theorem \ref{thm:hcube} is that, at time $t$, estimates from all agents at a distance $t$ from $v$ show up in an equal number of neighboring estimates, in this case $t$ of them. More generally, we can observe the following:  Let $G$ be a graph such that for any vertices $u,v,v'$ with $d(u,v) = d(u,v')$ there is an automorphism that fixes $u$ and swaps $v$ and $v'$. Then the process converges in $D$ steps on the graph $G$. Such an automorphism exists for all distance-transitive graphs \cite{wiki:dtg}.
\begin{definition}
	A graph $G$ is distance transitive if, given any two vertices $v$ and $w$ at any distance $i$, and any other two vertices $x$ and $y$ at the same distance, there is an automorphism of the graph that carries $v$ to $x$ and $w$ to $y$.
\end{definition}
We have thus in fact proved the following, more general theorem.
\begin{theorem}[Distance Transitive Graphs]\label{thm:distance-transitive}
		Beliefs converge in $D$ steps if $G$ is a distance transitive graph.
\end{theorem}

\section{Privacy}  
	In this section we define and analyze a notion of privacy.     The question the we ask is: given that the process aggregates information optimally, so that each agent learns the average of all the private signals, is it possible that agents do not learn much beyond that about each others' signals? We show that for some graphs this is indeed the case.
    
    \begin{definition}
    	Let $\Pi_{wv}$ be the normalized variance of an agent $w$'s estimate of another agent $v$'s private signal. Formally, $\Pi_{wv} = \frac{1}{\sigma^2}\text{Var}_w(S_v|\textbf{r})$, where $w$ is the agent making the estimate and $\textbf{r}$ is the vector of $w$'s observations at the time of convergence.
    \end{definition}

    Each element of $\Pi$ is in the range $[0,1]$, where $\Pi_{wv} = 0$ if $w$ can exactly compute $S_v$ from its observations and $\Pi_{wv} = 1$ when $w$ has not yet observed $v$. Clearly $\Pi_{ww} = 0$, since every agent knows its own private measurement. In addition, if $v$ is a neighbor of $w$ then $\Pi_{wv} = 0$, since each agent directly observes its neighbors.  
		
	Things become more complicated once we start looking at agents that are at a distance greater than 1. As a motivating example, we will consider social learning on the star graph $G = K_{1,3}$. We denote agent $0$ to be at the center of the star and agents $1,2$ and $3$ to be the leaves.
At $t=0$, agent $1$ observes only its private measurement $S_1$, so $X_1(0) =\textbf{r} = S_1$. At time $t = 1$, Agent $1$ receives the estimate $X_0(0) = S_0$ from its only neighbor, so $\textbf{r} = [S_0,S_1]$ and its estimate of the world is $X_1(1) = \frac{1}{2}(S_0 + S_1)$. Since agent $0$ has seen everyone's private measurement at $t=1$, its estimate is $X_0(1) = \frac{1}{4}(S_0 + S_1 + S_2 + S_3)$. At $t = 2$, agent $1$ observes $X_0(1)$ and so they have $\textbf{r} = [S_0,S_1,\frac{1}{4}(S_0 + S_1 + S_2 + S_3)]$. Since each agent knows the structure of the graph, agent $1$ (along with everyone else) knows that it has just received the optimal estimate. It now has $X_1(2) = \frac{1}{4}(S_0 + S_1 + S_2 + S_3)$, and the process is complete.
	
	From what we have seen, we know that we can rewrite $\textbf{r}$ at the last step as
    \begin{equation}\label{eq:example}
	\textbf{r} = A\textbf{S} =
    \begin{bmatrix}
	1 & 0 & 0 & 0 \\
	0 & 1 & 0 & 0 \\
	\frac{1}{4} & \frac{1}{4}  & \frac{1}{4}  & \frac{1}{4}
	\end{bmatrix}\begin{bmatrix}
	S_0 \\ S_1 \\ S_2 \\ S_3
	\end{bmatrix}.
	\end{equation}
	Intuitively, it seems that the best agent $1$ can do is to average $S_2$ and $S_3$, which can be computed from
	\begin{equation}\label{eq:est}
	2r_2 - \frac{1}{2}r_0 - \frac{1}{2}r_1 = \frac{1}{2}(S_2+S_3).
	\end{equation}
	Note that these coefficients can be computed from the right pseudoiverse of $A$. With this estimator, the level of privacy corresponds to 
	$$\Pi_{12} =  \frac{1}{\sigma^2}\text{Var}\left(S_2\Big|\frac{1}{2}(S_2 + S_3)\right) = \frac{1}{2}.$$
	
	To formalize this intuition, we recall the covariance matrix $C$ between the components of $\textbf{r}$ is $C = \sigma^2 AA^T$. What we need to be able to compute then is the covariance between $\textbf{r}$ and $\textbf{S}$. If we think of $C$ as a transformation of the Gaussian vector space spanned by the components of $\textbf{r}$, the corresponding transformation in the space spanned by the $S_v$ is 
	$$C' = A^{\dagger}CA^{\dagger T},$$
	where $A^{\dagger} = A^T(AA^T)^{-1}$ is the right pseudoinverse of $A$. 
	Now we show that the variance of the agent's estimator is the error corresponding to the projection from the subspace spanned by $\textbf{r}$ into the space spanned by $\textbf{S}$: .
	
	\begin{theorem}
		The privacy $\Pi$ can be computed from the equation $\Pi_{wv} = \text{diag}(I - C')_v$.
	\end{theorem}
	\begin{proof}
		 we note that the conditional variance of a multivariate Gaussian is
		 $$\text{Var}(X_1|X_2) = \Sigma_{11} - \Sigma_{12}\Sigma_{22}^{-1}\Sigma_{21}.$$
		 
		 Since we are computing $\text{Var}_w(\textbf{S}|\textbf{r})$, we know that 
		 \begin{eqnarray}
		 \Sigma_{11} &=& \sigma^2I, \\
		 \Sigma_{22} &=& C = \sigma^2AA^T, \nonumber \\
		 \Sigma_{12} &=& \Sigma_{21}^T =  \sigma^2 A^T \nonumber
		 \end{eqnarray}
		 Combining these, we have
		 \begin{eqnarray*}
		 	\text{Var}_w(\textbf{S}|\textbf{r}) &=& \Sigma_{11} - \Sigma_{12}\Sigma_{22}^{-1}\Sigma_{21} \\
		 	&=& \sigma^2I - \sigma^2A^T(AA^T)^{-1}A \\
		 	&=& \sigma^2(I - C')
		 \end{eqnarray*}
		 From this, the variance of a given estimator is
		 $$\Pi_{wv} = \frac{1}{\sigma^2}\text{Var}_w(\textbf{S}|\textbf{r})_{vv} = \text{diag}(I - C')_v.$$
	\end{proof}
	
	Returning to the example above, we have that 
	\begin{equation*}
	I - C' = 
	\begin{bmatrix}
	0 & 0 & 0 & 0 \\
	0 & 0 & 0 & 0 \\
	0 & 0 & \frac{1}{2} & -\frac{1}{2} \\
	0 & 0 & -\frac{1}{2} & \frac{1}{2}
	\end{bmatrix},
	\end{equation*}
	so the privacy of the estimators is 
	$$\Pi_{w,0:3} = \left[0,0,\frac{1}{2},\frac{1}{2}\right].$$
	Since agent 0 directly observed $S_0$ and $S_1$ and only observed the average of $S_2$ and $S_3$, this matches precisely what we would expect.
	
	The next step is to examine how the structure of the graph and the convergence rate relate to the notion of privacy we have discussed. One immediate result is that if any agent has a rank $N$ covariance matrix, then as mentioned earlier $A$ is full rank and invertible. If $A$ is invertible, then that agent can directly calculate $\textbf{S} = A^{-1}\textbf{r}$, and thus there is no privacy with respect to that agent. This result can also be seen from the above theorem, as an invertible $A$ tells us that
	$$I - A^T(AA^T)^{-1}A = I - I = 0.$$
	
	More generally, we can apply these results to the two families of graphs studied in the previous sections, trees and hypercubes.
    \begin{theorem}
    	The privacy between an agent $w$ in a tree and an agent $v$ at distance $k$ away is $\Pi_{wv} = 1-\frac{1}{|H^c_w(k)|}$, where $H^c_w(k)$ is the set of agents at distance exactly $k$ from $w$ in the subtree containing $v$.
    \end{theorem}
    \begin{proof}
    We know from theorem \ref{thm:trees} that $w$ is able to compute from its observations the average of all agents in $H^c_w(k)$, the set of agents exactly distance $k$ from $w$ in the subtree $G_c$. Since this is the only observation $w$ gets that has any new information about $v$. This means that 
   \begin{multline*}
   \Pi_{wv} = \text{Var}_w \left(S_v \Bigg |\frac{1}{|H^c_w(k)|}\sum_{v' \in H^c_w(k)}S_{v'} \right) \\= 1-\frac{1}{|H^c_w(k)|}.
   \end{multline*}
    \end{proof}
	For the special case of balanced $m$-ary trees, we see that the privacy between the root agent and any agent at depth $k$ is precisely $\text{Var}_w(S_v|\textbf{r}) = \sigma^2(1-\frac{1}{m^{k-1}})$.
    
    \begin{theorem}
    If $G$ is a hypercube, then the privacy between agents at distance $k$ apart is $\Pi_{wv} = 1 - \frac{n}{\binom{n}{k}}$.
    \end{theorem}
	\begin{proof}
		 We note that that, for an $n$-dimensional hypercube, $v$ has precisely $\binom{n}{k}$ agents at distance $k$ away. We also know that, at step $k$, $v$ received $n$ linearly independent observations from its neighbors that contain information about averages over subset of neighbors at distance $k$. From theorem \ref{thm:hcube} we know we can isolate just the terms 
         \small
        $$\frac{|B_v(t)|}{|H_v(t)|}x_v(t) - \frac{|B_v(t-1)|}{|H_v(t)|}x_v(t-1) = \frac{1}{|H_{v}(t)|}\sum_{w \in H_{v}(t)}S_w$$\normalsize
         Now let $A$ be the matrix whose rows only contain coefficients from these equations. From this, we can compute
    \small\begin{eqnarray*}
    	\text{Tr}(I - A^T(AA^T)^{-1}A) 
        &=& \text{Tr}(I) - \text{Tr}(A^T(AA^T)^{-1}A) \\
        &=& \binom{n}{k} - \text{Tr}(AA^T(AA^T)^{-1}) \\
        &=& \binom{n}{k} - n.
    \end{eqnarray*}\normalsize
    By symmetry, we observe that the privacy between an agent $w$ and any agent $v$ at distance $k$ on the hypercube must be the same for all $v$. Since there are $\binom{n}{k}$ agents at distance $k$, we know that
    $$\Pi_{wv} = \frac{1}{\binom{n}{k}}\text{Tr}(I - A^T(AA^T)^{-1}A) = 1 - \frac{n}{\binom{n}{k}}.$$
	\end{proof}
    
    This theorem demonstrates an interesting privacy pattern on this graph: privacy is zero for neighbors, is maximized near the middle of the graph (i.e., where $k = [\frac{n}{2}]$) and goes back down to zero for nodes of the other side of the graph (i.e., where $k=n$). This result also generalizes to distance transitive graphs. If $d(w,v) = k$, then we will have privacy equal to
    $$\Pi_{wv} = 1 - \frac{H_w(1)}{H_w(k)},$$
   i.e. one minus the ratio of the number of neighbors of $w$ to the number of agents at distance $k$ from $w$.
   
\section{Follow-up work and future research directions}
Since the first version of this paper~\cite{mossel2010efficient}
appeared online, it has influenced a number of other studies (see.,
e.g.,~\cite{golub2012homophily, jadbabaie2012non,
  corazzini2012influential, frongillo2011social, mossel2012agreement,
  eksinlearning, eksin2013learning}), some of which are direct
generalizations of this work. We discuss some of these and other
possible research directions below.

A property of this model that DeMarzo et al.\ found unrealistic is the
requirement that all agents know the structure of the social
network. While indeed this may be difficult to justify for some large
networks, {\em it is perhaps not strictly necessary}; in order to
perform their calculations, the agents need to know the covariance
between the estimators of {\em their neighbors only}.  In our model,
they derive this knowledge from the structure of the graph, but in
principle it may be derived by other means. This observation (which
was clarified in discussions with Rafael M.\ Frongillo, Grant
Schoenebeck and Adam Kalai, whom we would like to thank) presents an
opportunity for follow-up work involving some variant of this model,
which does not require knowledge of the network, but still preserves
rationality, tractability and efficient learning.

Another promising direction is to consider the case that the state of
the world is not constant but changes stochastically as the agents try
to estimate it. This is pursued in~\cite{frongillo2011social}, for the
case of the complete network.

A further natural generalization is the consideration of more
complicated utilities. For example, each agent's utility could depend
also on the actions of others. This is pursued
in~\cite{eksinlearning}.

Not much need be changed in this model to make the agents strategic.
For example, consider a model that is identical to ours, except that
agents are no longer myopic, but take an action that optimizes their
expected future discounted gains (i.e. the expectation of their
averaged future gains, with exponentially decreasing weights). Agents
will now take into account the impact of their actions on their peers,
and so it seems plausible that strategic behavior will emerge, and in
fact this is proved for similar models in~\cite{rosenberg:06}. Their
work leads us to conjecture that this modified non-myopic model may
feature strategic interactions on social networks with tractable
calculations. No other such example is known - to the best of our
knowledge.

The results established in this paper show convergence in $O(N \cdot
D)$.  A natural question is whether this bound can be
improved. Certainly, convergence cannot happen faster than $O(D)$ -
the time it takes information to propagate through the network. For
binary trees, where the diameter is $O(\log N)$, convergence does
happen in $O(D)$, as it does for cliques and distance-transitive graphs. However,
simulations have led us to believe that convergence in general is not
that fast, and requires - we conjecture - $O(N)$ steps.

In our simulations we sampled a population of $d$-regular graphs for a wide range of $d$. The result always converges in $N/d$ steps, with every agent increasing the dimension of the space
its memory spanned by $d$, at every iteration.  This may hint that
convergence time may, in some sense, be inversely proportional to the
degrees of the graph vertices. Additionally we find that distance-transitive graphs (which converged in $O(D)$) converge in $O(N/d)$ when a few pairs of edges are randomly swapped. This means that a highly symmetric graph will behave as though it were completely randomized under mild perturbations (usually two or three swaps will suffice).  We thus conclude with the following
conjecture and open problem:

\begin{conjecture}
For any graph the learning process converges in $O(N)$ iterations.
\end{conjecture}

\begin{open}
Does the process converge in $O(N/d^{\ast})$ iterations for all graphs, where $d^{\ast}$ is the minimal degree of the graph?
\end{open}

\section{Acknowledgments}
We thank Shachar Kariv for introducing us to the literature on learning on networks and for fascinating discussions.
Thanks go to Marcus M\"obius for an interesting discussion and for directing us
to the work of DeMarzo et al., after a draft of this note had been
circulated.
We also thank Grant Schoenebeck, Rafael M.\ Frongillo and Adam Kalai for
interesting discussions regarding follow-up work. We thank Ashwin Ganesan for a lesson on distance-transitive graphs. Finally, we are indebted
to Yaron Singer for much support and helpful suggestions on writing the results.

\bibliographystyle{IEEEtran} \bibliography{IEEEabrv,all}
\end{document}